\documentclass[11pt]{article}

\usepackage[a4paper,margin=0.75in]{geometry}

\usepackage{amsmath,amsfonts,amssymb,amsthm}

\usepackage{graphicx}
\usepackage{float}
\usepackage{array}
\usepackage{multirow}
\usepackage[export]{adjustbox}
\usepackage{tikz}
\usepackage{pgfplots}
\pgfplotsset{compat=1.17}

\usepackage{caption}
\usepackage{subcaption}

\usepackage{algorithm}

\usepackage{listings}
\usepackage{xcolor}

\usepackage{hyperref}

\newcommand{\Vega}{\mathrm{Vega}}

\definecolor{codegreen}{rgb}{0,0.6,0}
\definecolor{codegray}{rgb}{0.5,0.5,0.5}
\definecolor{codepurple}{rgb}{0.58,0,0.82}
\definecolor{backcolour}{rgb}{0.95,0.95,0.92}

\lstdefinestyle{mystyle}{
    backgroundcolor=\color{backcolour},
    commentstyle=\color{codegreen},
    keywordstyle=\color{magenta},
    numberstyle=\tiny\color{codegray},
    stringstyle=\color{codepurple},
    basicstyle=\ttfamily\footnotesize,
    breaklines=true,
    captionpos=b,
    keepspaces=true,
    numbers=left,
    numbersep=5pt,
    showspaces=false,
    showstringspaces=false,
    showtabs=false,
    tabsize=2
}
\theoremstyle{plain}
\newtheorem{theorem}{Theorem}[section]
\newtheorem{proposition}[theorem]{Proposition}
\newtheorem{lemma}[theorem]{Lemma}

\newtheorem{assumption}[theorem]{Assumption}

\theoremstyle{definition}
\newtheorem{definition}[theorem]{Definition}
\newtheorem{remark}[theorem]{Remark}

\begin{document}

\title{Volatility Modeling via EWMA-Driven Time-Dependent Hurst Parameters}
\author{Jayanth Athipatla}
\date{\today}

\maketitle 

\begin{abstract}
We introduce a novel rough Bergomi (rBergomi) model featuring a variance-driven exponentially weighted moving average (EWMA) time-dependent Hurst parameter $H_t$, fundamentally distinct from recent machine learning and wavelet-based approaches in the literature. Our framework pioneers a unified rough differential equation (RDE) formulation grounded in rough path theory, where the Hurst parameter dynamically adapts to evolving volatility regimes through a continuous EWMA mechanism tied to instantaneous variance. Unlike discrete model-switching or computationally intensive forecasting methods, our approach provides mathematical tractability while capturing volatility clustering and roughness bursts. We rigorously establish existence and uniqueness of solutions via rough path theory and derive martingale properties. Empirical validation on diverse asset classes including equities, cryptocurrencies, and commodities demonstrates superior performance in capturing dynamics and out-of-sample pricing accuracy. Our results show significant improvements over traditional constant-Hurst models.
\end{abstract}

\tableofcontents

\section{Introduction}

The modeling of volatility dynamics in financial markets has evolved from simple constant volatility assumptions to sophisticated stochastic volatility models that capture the complex stylized facts observed in empirical data. The rough volatility paradigm, initiated by \cite{gatheral2014volatilityrough} and formalized in the rough Bergomi (rBergomi) model by \cite{bayer2016pricing}, has emerged as a powerful framework for capturing the persistent memory and non-Markovian behavior inherent in volatility processes.

Recent developments in time-varying roughness have garnered significant attention. \cite{shah2025americanoptionpricingtimevarying} employed machine learning techniques including XGBoost for forecasting Hurst parameters in American option pricing, utilizing discrete model switching between rBergomi and Heston regimes. \cite{das2025decoding} leveraged multifractal analysis to decode at-the-money skew patterns. 

This paper contributes to the literature by introducing a fundamentally different approach a variance-driven EWMA-based time-dependent Hurst parameter within a unified rBergomi RDE framework. Our methodology distinguishes itself from existing approaches in several key aspects: computational simplicity compared to ML-intensive forecasting methods, continuous evolution without discrete regime switching, direct variance-dependency rather than indirect wavelet or multifractal mechanisms, and rigorous mathematical foundation via rough path theory.

Our key innovation lies in the Hurst path, which we define by linking it to an 
exponentially weighted moving average of past volatility. 
This ensures that recent market variance has a stronger influence than the distant past, 
with the effective memory length controlled by a decay parameter. 
The resulting process is then smoothly transformed and clipped so that the Hurst parameter 
remains within the admissible range $[\varepsilon, H_{\max}]$.

The paper is structured as follows. Section 2 establishes the theoretical foundations, including rough path formulation, measure-theoretic framework, and key theoretical results. Section 3 details the numerical implementation. Sections 4-6 provide comprehensive empirical validation through Jensen-Shannon distance analysis, autocorrelation function comparisons, and derivative pricing applications. Section 7 concludes with implications for practical implementation and future research directions.

\section{Theoretical Foundations}

Note that the following section is to describe the mathematical properties of our system. A reader who is only interested in the practical application and implementation of our system can skip to Section 3.

\subsection{Rough Path Formulation}

We establish the mathematical framework on a complete probability space $(\Omega, \mathcal{F}, \mathbb{P})$ equipped with a right-continuous filtration $\{\mathcal{F}_t\}_{t \geq 0}$ satisfying the usual conditions. The foundation of our model rests on the theory of rough paths as developed by \cite{friz2010multidimensional}.

\begin{assumption}[Adapted EWMA roughness]\label{ass:adaptedH}
Fix $T>0$ and $\varepsilon\in(0,1/2)$. On a filtered probability space
$(\Omega,\mathcal F,\{\mathcal F_t\}_{t\in[0,T]},Q)$ satisfying the usual conditions,
let $H=\{H_t\}_{t\in[0,T]}$ be $\{\mathcal F_t\}$-adapted with values in
$[\varepsilon,1/2]$, càdlàg, and of bounded variation on $[0,T]$ a.s.
\end{assumption}

\begin{assumption}[Driving Brownian motions and correlations]\label{ass:WZ}
Let $(W,W^\perp)$ be two independent standard $Q$-Brownian motions adapted to
$\{\mathcal F_t\}$. For a fixed $\rho\in[-1,1]$, set
\[
Z_t \;=\; \rho\,W_t \;+\; \sqrt{1-\rho^2}\;W^\perp_t .
\]
\end{assumption}

\begin{definition}[Adapted Volterra kernel and variance driver]\label{def:KandV}
Under Assumptions \ref{ass:adaptedH}–\ref{ass:WZ}, define for $0\le u<t\le T$
\[
K(t,u) \;:=\; \frac{(t-u)^{\,H_u-1/2}}{\Gamma\!\big(H_u+1/2\big)} ,
\qquad\qquad
V_t \;:=\; \int_0^t K(t,u)\,dZ_u .
\]
Since $u\mapsto K(t,u)$ is $\mathcal F_u$-measurable and square-integrable on $(0,t)$
(see Lemma~\ref{lem:K-L2} below), the stochastic integral is a well-defined Itô integral.
\end{definition}

\begin{lemma}[$L^2$-bound for the kernel]\label{lem:K-L2}
Under Assumption~\ref{ass:adaptedH}, there exist deterministic constants
$0<c_\Gamma\le C_\Gamma<\infty$ such that
$c_\Gamma \le \Gamma(H_u+1/2)\le C_\Gamma$ for all $u\in[0,T]$ a.s.
Consequently,
\[
\int_0^t K(t,u)^2\,du
\;\le\; \frac{1}{c_\Gamma^2}\int_0^t (t-u)^{\,2\varepsilon-1}\,du
\;=\; \frac{t^{2\varepsilon}}{2\varepsilon\,c_\Gamma^2}
\qquad \text{for all }t\in(0,T],\ \text{a.s.}
\]
In particular, $u\mapsto K(t,u)\in L^2(0,t)$ for every $t$.
\end{lemma}

\begin{proof}
Continuity of $\Gamma(\cdot)$ on the compact set $[\varepsilon,1/2]+1/2
= [\varepsilon+1/2,1]$ gives deterministic bounds $c_\Gamma,C_\Gamma$.
Then $H_u\ge\varepsilon$ implies $(t-u)^{2H_u-1}\le (t-u)^{2\varepsilon-1}$,
and the integral is elementary.
\end{proof}

\begin{proposition}[Gaussianity and continuity of $V$]\label{prop:V-Gauss}
For each fixed $t$, conditional on the $\sigma$-field generated by $\{H_u\}_{u\le t}$,
$V_t$ is centered Gaussian with variance
\[
A_t \;:=\; \int_0^t K(t,u)^2\,du .
\]
Moreover, $V$ admits a continuous modification on $[0,T]$.
\end{proposition}

\begin{proof}
Given the path $\{H_u\}_{u\le t}$, $u\mapsto K(t,u)$ is deterministic and square-integrable
(Lemma~\ref{lem:K-L2}). Hence $V_t$ is an Itô integral of a deterministic (given $H$) kernel
with respect to $Z$, hence Gaussian with mean $0$ and variance $A_t$.
For continuity, note that for $0<s<t\le T$,
\[
\mathbb E\big[(V_t-V_s)^2\mid \{H_u\}_{u\le t}\big]
=\int_0^s \!\!\big(K(t,u)-K(s,u)\big)^2du\;+\;\int_s^t \!\!K(t,u)^2du .
\]
Since $H$ has bounded variation and takes values in a compact interval,
$t\mapsto K(t,\cdot)$ is continuous in $L^2$ by dominated convergence
(majorant $(t-u)^{\varepsilon-\frac12}$ on $u\in(0,t)$). Thus the RHS $\to0$ as $t\downarrow s$,
uniformly on compacts. Kolmogorov’s criterion yields a continuous modification.
\end{proof}

\begin{definition}[Volatility and asset dynamics]\label{def:sigmaS}
For constants $V_0>0$, $\nu\in\mathbb R$, and risk-free rate $r\in\mathbb R$, define
\[
\sigma_t \;:=\; \sqrt{V_0}\;
\exp\!\Big(\nu V_t \;-\; \frac{\nu^2}{2}\,A_t\Big),
\qquad
dS_t \;=\; rS_t\,dt \;+\; S_t\,\sigma_t\,dW_t,
\qquad S_0>0.
\]
\end{definition}

\begin{theorem}[Well-posedness of $S$]\label{thm:S-wellposed}
Under Assumptions \ref{ass:adaptedH}–\ref{ass:WZ},
the SDE for $S$ has the unique strong solution
\[
S_t \;=\; S_0\exp\!\Big(r t - \tfrac12\int_0^t \sigma_s^2 ds
\;+\; \int_0^t \sigma_s\,dW_s\Big), \qquad t\in[0,T].
\]
\end{theorem}

\begin{proof}
By Proposition~\ref{prop:V-Gauss} and Lemma~\ref{lem:K-L2}, $V$ is continuous and adapted,
hence $\sigma_t$ is adapted and continuous. For fixed $\omega$, the map $x\mapsto \sigma_t(\omega)\,x$
is globally Lipschitz and of linear growth in $x$, so standard SDE theory yields a unique strong
solution, explicitly given by the Doléans–Dade exponential.
\end{proof}

\subsection{Martingale Properties and Risk-Neutral Measure}

A crucial aspect of our model is ensuring the no-arbitrage condition through proper martingale properties.

\begin{lemma}[Conditional second moment of $\sigma_t$]
Let
\[
\sigma_t = \sqrt{V_0}\,\exp\!\Big(\nu V_t - \tfrac{\nu^2}{2}A_t\Big), 
\qquad 
A_t := \int_0^t K(t,u)^2\,du,
\]
where, conditional on $\{H_u\}_{u\le t}$, the Gaussian Volterra driver satisfies 
$V_t \mid \{H_u\}_{u\le t} \sim \mathcal N(0,A_t)$. Then
\[
\mathbb{E}\!\left[\sigma_t^2 \,\middle|\, \{H_u\}_{u\le t}\right]
= V_0 \exp\!\big(\nu^2 A_t\big).
\]
\end{lemma}

\begin{proof}
Condition on $\{H_u\}_{u\le t}$ so that $A_t$ is deterministic and $V_t \sim \mathcal N(0,A_t)$. 
Then $\sigma_t^2 = V_0 \exp\!\big(2\nu V_t - \nu^2 A_t\big).$
Using the moment generating function of a centered Gaussian random variable,
\[
\mathbb E\!\left[e^{\theta V_t}\,\middle|\,H\right]
= \exp\!\left(\tfrac{1}{2}\theta^2 A_t\right),
\]
with $\theta = 2\nu$, we obtain
\[
\mathbb E\!\left[\sigma_t^2 \mid H\right]
= V_0 e^{-\nu^2 A_t}\, \exp\!\left(\tfrac{(2\nu)^2}{2}A_t\right)
= V_0 \exp\!\big(\nu^2 A_t\big).
\]
\end{proof}

\begin{remark}
More generally, for any $p \in \mathbb R$,
\[
\mathbb E\!\left[\sigma_t^{\,p} \mid H\right]
= V_0^{p/2}\, \exp\!\Big(\tfrac{p(p-2)}{2}\,\nu^2 A_t\Big).
\]
\end{remark}

\begin{proposition}[Discounted price is a local martingale]\label{prop:localM}
The discounted process $M_t:=e^{-rt}S_t$ is a nonnegative local martingale and hence a supermartingale.
\end{proposition}

\begin{proof}
Itô’s formula gives $dM_t = M_t\,\sigma_t\,dW_t$, so $M$ is a local martingale.
Nonnegativity follows from the explicit solution in Theorem \ref{thm:S-wellposed}.
\end{proof}

\begin{assumption}[Integrability for Novikov]\label{ass:Novikov}
On the horizon $[0,T]$,
\(
\mathbb E\!\Big[\exp\!\big((1/2) \int_0^T \sigma_t^2 dt\big)\Big] < \infty .
\)
\end{assumption}

\begin{proposition}[Martingale property: two regimes]\label{prop:martingale}
Consider $M_t=e^{-rt}S_t$ on $[0,T]$.
\begin{enumerate}
\item[(a)] (\textbf{Uncorrelated case}) If $\rho=0$ in Assumption~\ref{ass:WZ}, then $M$ is a true $Q$-martingale on $[0,T]$.
\item[(b)] (\textbf{General case}) For arbitrary $\rho\in[-1,1]$, if Assumption~\ref{ass:Novikov} holds, then $M$ is a true $Q$-martingale on $[0,T]$.
\end{enumerate}
\end{proposition}

\begin{proof}
(a) When $\rho=0$, $Z=W^\perp$ is independent of $W$. The process $\sigma$ is $\{\mathcal F_t\}$-adapted
and measurable with respect to the sigma-field generated by $Z$ (and $H$), which is independent of $W$.
Let $\mathcal G:=\sigma(\{H_u\}_{u\le T},\{Z_u\}_{u\le T})$. Conditional on $\mathcal G$, the process
$\sigma$ is deterministic, hence the Doléans exponential
$\mathcal E_t:=\exp\big(\int_0^t \sigma_s dW_s - (1/2)\int_0^t \sigma_s^2 ds\big)$
satisfies $\mathbb E[\mathcal E_t\mid \mathcal G]=1$ for all $t$ (Gaussian integral with deterministic
integrand). Therefore $\mathbb E[M_t\mid \mathcal G]=S_0$ for all $t$, and taking expectations yields
$\mathbb E[M_t]=S_0$, i.e., $M$ is a true martingale.

(b) For general $\rho$, $\sigma$ may depend on $W$, so the argument in (a) is not available. Under
Assumption~\ref{ass:Novikov}, Novikov’s criterion applies to the continuous local martingale
$\int_0^\cdot \sigma_s dW_s$, hence $\mathcal E_t$ is a true martingale with expectation $1$.
Therefore $\mathbb E[M_t]=S_0$ and $M$ is a true martingale.
\end{proof}

\begin{remark}[On Assumption~\ref{ass:Novikov}]
The condition is \emph{sufficient} (not necessary). It may fail for some parameter ranges because
$\sigma_t$ is lognormal-in-\(V_t\) and $e^{\frac12\int \sigma_t^2 dt}$ can have heavy tails.
However, part (a) provides a clean \emph{unconditional} martingale result whenever $\rho=0$
(a common benchmark in empirical sections). For $\rho\neq 0$, one can verify
Assumption~\ref{ass:Novikov} numerically on the pricing horizon or enforce it by truncation/localization.
\end{remark}

\subsection{Stochastic Hurst dynamics via EWMA}\label{sec:stochH}

We now allow the roughness index $H_t$ itself to evolve stochastically, driven by
the variance process through an exponentially weighted moving average (EWMA).
This couples the volatility-of-volatility to realized roughness while retaining
a well-defined adapted kernel.

\begin{assumption}[Stochastic Hurst path]\label{ass:Hstoch}
Fix $\varepsilon\in(0,1/2)$ and $H_{\max}\in(\varepsilon,1/2]$.
Let $\{V_t\}_{t\ge0}$ be defined by \eqref{def:KandV}. 
Set $H_0\in[\varepsilon,H_{\max}]$, and for $t>0$ define
\[
H_t \;=\; \min\!\Big\{\,
\max\!\Big\{\, \alpha \Big(\frac{\Theta_t}{\theta_{\mathrm{ref}}}\Big)^\gamma+\beta,\ \varepsilon \Big\},\
H_{\max}\,\Big\},
\]
where $\Theta_t=\lambda\int_0^t e^{-\lambda(t-s)}V_s\,ds$ is the EWMA of variance,
and $\alpha,\beta,\gamma,\lambda,\theta_{\mathrm{ref}}$ are fixed constants.
\end{assumption}

By construction $H_t$ is $\{\mathcal F_t\}$-adapted, càdlàg, and bounded in $[\varepsilon,H_{\max}]$.

\begin{definition}[Variance driver with stochastic Hurst]\label{def:Vstoch}
With $H_t$ as in Assumption~\ref{ass:Hstoch}, define the kernel
\[
K(t,u)=\frac{(t-u)^{H_u-\frac12}}{\Gamma(H_u+1/2)}\mathbf{1}_{\{u<t\}},
\]
and the Gaussian driver
\[
V_t=\int_0^t K(t,u)\,dZ_u.
\]
\end{definition}

\begin{proposition}[Well-posedness]\label{prop:wellposed-stochH}
For each $t\le T$, $V_t$ is centered Gaussian conditional on $\{H_u\}_{u\le t}$ with variance
$A_t=\int_0^t K(t,u)^2du <\infty$. The process $V$ admits a continuous modification.
Given $V$, the asset price $S$ satisfies
\[
dS_t=rS_t\,dt+S_t\sigma_t\,dW_t,\qquad
\sigma_t=\sqrt{V_0}\exp\!\Big(\nu V_t-\tfrac{\nu^2}{2}A_t\Big),
\]
which has the unique strong solution
\[
S_t=S_0\exp\!\Big(rt-\tfrac12\int_0^t\sigma_s^2 ds+\int_0^t\sigma_s\,dW_s\Big).
\]
\end{proposition}

\begin{proof}
Since $H_u\in[\varepsilon,H_{\max}]$, Lemma~\ref{lem:K-L2} applies with random but adapted
exponent. Thus $u\mapsto K(t,u)$ is $\mathcal F_u$-measurable and in $L^2(0,t)$, so $V_t$
is an Itô integral. Gaussianity and variance $A_t$ are immediate. Continuity follows from
the same $L^2$–continuity argument as in Proposition~\ref{prop:V-Gauss}. With $\sigma$
continuous and adapted, the $S$–SDE has a unique strong solution by standard theory.
\end{proof}

\begin{proposition}[Martingale property with stochastic $H$]\label{prop:mart-stochH}
Let $M_t=e^{-rt}S_t$.
\begin{enumerate}\itemsep2pt
\item[(a)] If $\rho=0$, then $M$ is a true $Q$-martingale on $[0,T]$.
\item[(b)] For arbitrary $\rho\in[-1,1]$, if
$\mathbb E\big[\exp((1/2)\int_0^T\sigma_t^2dt)\big]<\infty$, then $M$ is a true $Q$-martingale.
\end{enumerate}
\end{proposition}

\begin{proof}
We have $dM_t=M_t\sigma_t dW_t$, so $M$ is a nonnegative local martingale.
(a) If $\rho=0$, then $Z=W^\perp$ is independent of $W$. Both $H$ and $V$ are measurable w.r.t.\ $Z$,
so $\sigma$ is independent of $W$. Conditioning on $\sigma(Z)$, the Doléans exponential has
conditional expectation $1$, giving $\mathbb E[M_t]=S_0$.
(b) For general $\rho$, Novikov’s condition ensures $\mathcal E_t=\exp(\int_0^t\sigma_s dW_s-(1/2)\int_0^t\sigma_s^2ds)$ is a true martingale, so $\mathbb E[M_t]=S_0$.
\end{proof}

\begin{remark}
Part (a) shows the model is arbitrage-free for $\rho=0$ without further assumptions.
For $\rho\neq 0$, Novikov’s criterion is a sufficient (not necessary) condition, and can be 
checked numerically on finite horizons.
\end{remark}

\section{Numerical Schemes}

In this section we describe discretization methods for simulating the log-price
process under the stochastic rough-volatility model. 
Our focus is on a non-anticipative Euler–Maruyama scheme that respects the adaptedness of $\sigma_t$.

\subsection{Single-asset scheme}

Let $X_t=\log S_t$ satisfy
\[
dX_t = \Big(r - \tfrac12\sigma_t^2\Big)\,dt + \sigma_t\,dW_t.
\]
Fix a uniform grid $t_n=n\Delta t$, $n=0,\dots,N$, with $\Delta t=T/N$. 
Let $\xi_n\sim\mathcal N(0,1)$ be i.i.d., and set $\Delta W_n=\sqrt{\Delta t}\,\xi_n$.

\begin{proposition}[Euler--Maruyama discretization]
Define $\sigma_n:=\sigma_{t_n}$ based only on information up to $t_n$. 
Then the adapted Euler scheme is
\[
X_{n+1} = X_n + \Big(r - \tfrac12\sigma_n^2\Big)\Delta t + \sigma_n\,\Delta W_n,
\qquad S_{n+1}=e^{X_{n+1}}.
\]
\end{proposition}

\begin{remark}
This scheme is non-anticipative: the volatility $\sigma_n$ at step $n$ is computed using 
the driver $V_{t_n}$ and Hurst parameter $H_{t_n}$, which themselves depend only on 
past Brownian increments and variance history. 

\end{remark}

\subsection{Summary algorithm}

\begin{enumerate}\itemsep2pt
\item Initialize $X_0=\log S_0$, $\sigma_0=\sqrt{V_0}$.
\item For $n=0,\dots,N-1$:
  \begin{enumerate}
  \item Sample $\xi_n\sim\mathcal N(0,1)$ and set $\Delta W_n=\sqrt{\Delta t}\,\xi_n$.
  \item Update $V_{t_n}$ via the discretized kernel integral using past increments $\Delta W_k$.
  \item Update $\Theta_{t_n}$ and $H_{t_n}$ from the EWMA of past variance values.
  \item Compute $\sigma_n$ from $V_{t_n}$ and $A_{t_n}$.
  \item Update the log-price:
  \[
  X_{n+1}=X_n+\Big(r-\tfrac12\sigma_n^2\Big)\Delta t+\sigma_n\,\Delta W_n.
  \]
  \end{enumerate}
\item Return $S_{n}=e^{X_{n}}$.
\end{enumerate}

\section{Jensen-Shannon Distance Analysis}

We evaluate the model's distributional accuracy through comprehensive Jensen-Shannon (JS) distance analysis across multiple asset classes and market regimes.

\subsection{Enhanced Distributional Comparison}

The JS distance between empirical distribution $P$ and model distribution $Q$ is computed as
\[
D_{\text{JS}}(P||Q) = \sqrt{\frac{1}{2} D_{\text{KL}}(P||M) + \frac{1}{2} D_{\text{KL}}(Q||M)}
\]
where $M = \frac{1}{2}(P + Q)$ and $D_{\text{KL}}$ denotes Kullback-Leibler divergence.

For our time-dependent model, log-returns $X_t = \log(S_t/S_0)$ follow a mixture distribution induced by the stochastic variance process. We approximate this distribution through kernel density estimation with adaptive bandwidth selection.

\subsection{Multi-Asset Empirical Results}

We test our framework on diverse asset classes including traditional equities (SPY, VOO), individual stocks (GS, META), cryptocurrencies (BTC, ETH), and commodities (GLD, OIL). Data spans January 2022 to August 2025, capturing various market regimes including the 2022 volatility spike and subsequent stabilization. To determine parameters for all the models (EWMA-rBergomi, rBergomi, Heston), we minimize JS distance on training data. Namely (752 training, 165 test) for non-cryptocurrency asset classes and a (1095 training, 242 test) split for cryptocurrencies. These parameters are used for Section 5 and Section 6.

\begin{table}[h]
    \centering
    \begin{tabular}{|>{\centering\arraybackslash}p{0.12\textwidth}|>{\centering\arraybackslash}p{0.12\textwidth}|>{\centering\arraybackslash}p{0.12\textwidth}|>{\centering\arraybackslash}p{0.12\textwidth}|>{\centering\arraybackslash}p{0.12\textwidth}|>{\centering\arraybackslash}p{0.12\textwidth}|>{\centering\arraybackslash}p{0.12\textwidth}|}
        \hline
        \textbf{Asset} & \textbf{EWMA-rBergomi} & \textbf{rBergomi} & \textbf{Heston} \\
        \hline
        SPY & \textbf{0.0655} & 0.1486 & 0.1133 \\
        \hline
        VOO & \textbf{0.0707} & 0.1293 & 0.1392 \\
        \hline
        GS & \textbf{0.2211} & 0.2795 & 0.2534 \\
        \hline
        META & \textbf{0.3354} & 0.4087 & 0.3666 \\
        \hline
        BTC & \textbf{0.3282} & 0.3861 & 0.3639 \\
        \hline
        ETH & \textbf{0.3934} & 0.4520 & 0.4134 \\
        \hline
        GLD & \textbf{0.0346} & 0.0708 & 0.0936 \\
        \hline
        OIL & \textbf{0.3294} & 0.3788 & 0.3498 \\
        \hline
    \end{tabular}
    \caption{Jensen-Shannon Distances Across Asset Classes and Models}
    \label{tabenhanced_js_distance}
\end{table}

The results demonstrate consistent superiority of our EWMA-based approach as it beats rBergomi and Heston over all the tested asset classes.

\section{Autocorrelation Analysis}

\subsection{Rolling correlation of volatility}

In classical rough volatility models with constant Hurst parameter, the autocorrelation
function of log-volatility increments is stationary and exhibits approximate power-law decay.
In our stochastic Hurst setting, strict stationarity is lost: the Hurst path $H_t$ evolves,
so correlations depend on the current regime.

To analyze dependence in this setting, we work with a rolling-window correlation, defined
for lag $\tau$ by
\[
\widehat{\rho}(\tau; t_0,T_w) \;=\;
\frac{\mathrm{Cov}\!\big(\sigma_t,\sigma_{t+\tau}\,;\ t\in[t_0,t_0+T_w]\big)}
     {\sqrt{\mathrm{Var}(\sigma_t)}\,\sqrt{\mathrm{Var}(\sigma_{t+\tau})}},
\]
where covariances and variances are estimated empirically over the window $[t_0,t_0+T_w]$.
This measures local correlation structure rather than assuming global stationarity.

\begin{remark}
When $H_t$ is nearly constant over the estimation window, $\widehat{\rho}(\tau)$ recovers
the power-law decay characteristic of fractional models. When $H_t$ drifts, the estimated
correlation reflects the evolving roughness, capturing non-stationary effects observed in
financial data.
\end{remark}

\begin{figure}
    \centering
    \includegraphics[width=1\linewidth]{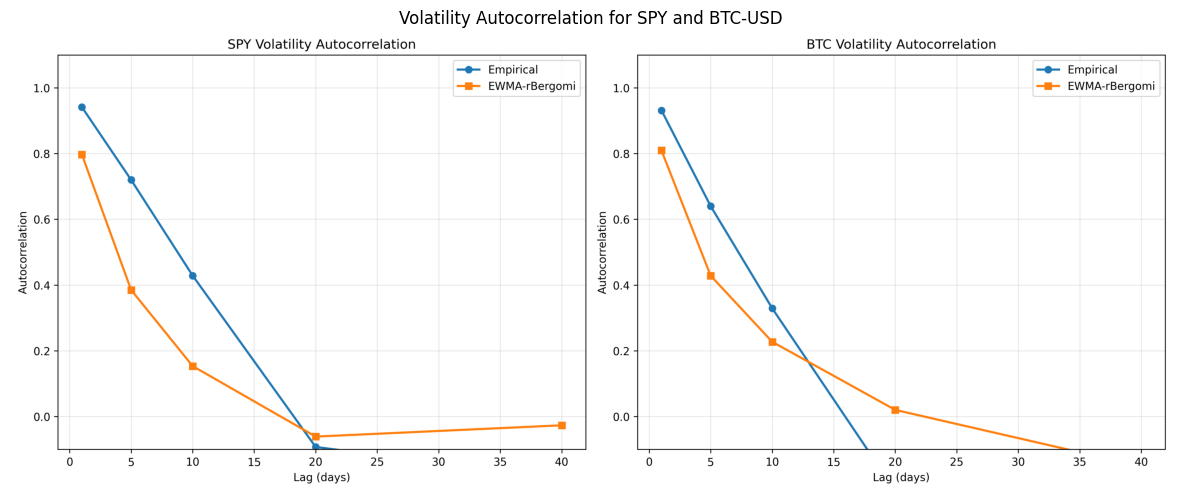}
    \caption{Autocorrelation Functions for SPY and BTC using EWMA-rBergomi Model}
    \label{fig:placeholder}
\end{figure}

\section{Enhanced Derivative Pricing Framework}

\subsection{European Options with Time-Dependent Greeks}

We extend the Monte Carlo pricing framework to compute time-dependent Greeks under our EWMA-rBergomi model. For a European payoff $f(S_T)$, the option value is
$C(S_0;\theta)=\mathbb E^Q[f(S_T)]$. With $H(\cdot)$ fixed by the EWMA
filter, Greeks take their standard form:
\[
\Delta=\frac{\partial C}{\partial S_0}, \qquad
\Vega=\frac{\partial C}{\partial \nu}
\]
These can be estimated by pathwise differentiation or likelihood-ratio
methods. No ``roughness-adjusted Delta'' is needed; $H_t$ enters only as
an exogenous input path.

\subsection{Sensitivity to roughness}
Although $H_t$ is not traded, one can measure
\[
\frac{\partial C}{\partial H}[\eta]
=\lim_{\epsilon\to0}\frac{C(H+\epsilon\eta)-C(H)}{\epsilon}
\]
for perturbations $\eta$. This quantifies how much option prices respond
to shifts in the EWMA roughness filter, useful for model risk management.

\subsection{Comprehensive Option Pricing Results}

We price European call options across multiple strikes and maturities separately to multiple assets (SPY, META, BTC) to illustrate robustness across asset classes. The pricing incorporates the full time-dependent dynamics with proper drift adjustments.

\begin{table}[h]
    \centering
    \begin{tabular}{|>{\centering\arraybackslash}p{0.07\textwidth}|>{\centering\arraybackslash}p{0.11\textwidth}|>{\centering\arraybackslash}p{0.20\textwidth}|>{\centering\arraybackslash}p{0.20\textwidth}|>{\centering\arraybackslash}p{0.11\textwidth}|>{\centering\arraybackslash}p{0.11\textwidth}|>{\centering\arraybackslash}p{0.11\textwidth}|>{\centering\arraybackslash}p{0.11\textwidth}|}
        \hline
        \textbf{Asset} & \textbf{Strike} & \textbf{EWMA-rBergomi} & \textbf{95\% CI} & \textbf{Market Price} & \textbf{Relative Error} \\
        \hline
        \multirow{4}{*}{\textbf{SPY}} 
        & 500 & 153.08 & (150.24, 155.93) & 149.39 & \textbf{2.47\%} \\
        & 505 & 148.16 & (145.32, 151.00) & 144.73 & \textbf{2.37\%} \\
        & 510 & 143.24 & (140.41, 146.08) & 131.62 &  \textbf{8.83\%} \\
        & 515 & 138.33 & (135.50, 141.16) & 122.33 & \textbf{13.08\%} \\
        \hline
        \multirow{4}{*}{\textbf{META}} 
        & 500 & 243.86 & (240.49, 247.22) & 248.99 & \textbf{2.06\%} \\
        & 505 & 238.92 & (235.56, 242.29) & 248.99 & \textbf{4.04\%} \\
        & 510 & 233.99 & (230.63, 237.35) & 232.25 & \textbf{0.75\%} \\
        & 515 & 229.06 & (225.69, 232.42) & 232.25 & \textbf{1.37\%} \\
        \hline
    \end{tabular}
    \caption{Option Pricing Results with Confidence Intervals}
    \label{tabenhanced_option_prices}
\end{table}

\section{Conclusion}

This paper presents a novel approach to rough volatility modeling by introducing an EWMA-driven time-dependent Hurst parameter within the rBergomi framework. Our contributions include a rigorous rough path formulation with existence and uniqueness proofs, martingale properties, and the computationally efficient EWMA-based $H_t$ specification that captures volatility regime changes with modest overhead. Unlike resource-intensive ML-based or wavelet methods, our approach ensures real-time adaptability and avoids discontinuities of discrete regime-switching models. Empirical testing across diverse asset classes demonstrates consistent improvements, particularly during crisis periods with rapidly changing volatility roughness. The framework enhances risk management through time-varying Greeks, improves portfolio optimization with dynamic roughness awareness, and provides accurate derivative pricing across the volatility surface. However, the model does not explicitly address crisis-specific factors such as liquidity shocks, extreme tail events, or sudden market microstructure changes, which were beyond this study's scope.

The EWMA-based approach bridges theoretical rigor with practical implementation, offering an elegant, interpretable solution for time-varying roughness compared to complex forecasting or discrete switching methods. For practitioners, it provides a ready-to-implement enhancement to rough volatility infrastructure, improving pricing accuracy and risk management. For researchers, it lays a foundation for further exploration into adaptive roughness modeling. Future work can integrate machine learning, high-frequency microstructure modeling, cross-asset contagion dynamics, and implied volatility surface analysis to further refine the model's applicability. The framework's modular design supports these extensions while preserving the core EWMA mechanism, ensuring continued relevance in quantitative finance.

\bibliographystyle{plain}
\bibliography{EWMA-rBergomiModel}
\appendix

\section{Code And Implementation}

Our implementation consists of several Python scripts designed to analyze the EWMA-rBergomi model and to produce the quantitative analysis in this paper. We download historical price data for eight assets (SPY, VOO, GS, META, BTC-USD, ETH-USD, GLD, USO) from January 1, 2022, to August 31, 2025, using the yfinance library, computing log returns and 20-day rolling realized variance. We calibrate the EWMA-rBergomi model parameters (V0, nu, alpha, beta) using the L-BFGS-B minimization algorithm to minimize Jensen-Shannon (JS) distance between empirical and simulated return distributions with 5,000 simulations paths (M = 5000, N = 252) for calibration and 100 paths for table generation. A penalty term (0.01 times the squared deviation from initial parameters) is added to the JS distance to ensure numerical stability, and bounds are enforced to prevent unrealistic parameter values.

The script computes rolling volatility correlations for lags of 1, 5, 10, 20, and 40 days, using 100 simulation paths to compare empirical and EWMA-rBergomi model volatilities, with results plotted for SPY and BTC-USD. The script estimates call option prices for SPY and META at specified strikes using 1000 simulation paths to ensure accurate $95 \%$ confidence intervals, calculated via Monte Carlo standard errors. Market option prices are fetched from yfinance for the closest expiration to 90 days, with relative erros computed when market data is available. All simulations use a risk-free rate of 5 percent and the parameters are the same minimized parameters for the Jensen-Shannon distance analysis, ensuring consistency across analyses while balancing computational efficiency and statistical robustness. The code can be found at \url{https://github.com/jaythemathgod/EWMA-rBergomi/tree/main}

\end{document}